\documentclass[letterpaper, 10 pt, conference]{ieeeconf}  

\IEEEoverridecommandlockouts                              

\usepackage{graphicx}
\usepackage{amsmath}
\usepackage{amsfonts}
\usepackage{amssymb}
\usepackage{subcaption}
\usepackage{float}
\usepackage{multirow}
\usepackage{nicefrac}

\usepackage{algorithm}
\usepackage{algpseudocode}
\renewcommand{\algorithmicrequire}{\textbf{Input:}}
\renewcommand{\algorithmicensure}{\textbf{Output:}}

\usepackage[inline]{trackchanges}
\addeditor{TB}
\addeditor{NJ}
\addeditor{HP}

\newcommand{\RN}{{\rm I\!R}}

\newtheorem{prop}{Proposition}

\newtheorem{rem}{Remark}
\newtheorem{ass}{Assumption}

\newtheorem{prob}{Problem}

\title{\LARGE \bf
A Set-Based Approach for Robust Control Co-Design
}

\author{Trevor J. Bird$^{1}$, Jacob A. Siefert$^{2}$, Herschel C. Pangborn$^{2}$, and Neera Jain$^{1}$ 
\thanks{*This work is supported by the U.S. Office of Naval Research Thermal Science and Engineering Program under contract number N00014-21-1-2352.}
\thanks{$^{1}$Trevor J. Bird and Neera Jain are with the School of Mechanical Engineering, Purdue University, West Lafayette, IN 47907 USA, {\tt\small bird6@purdue.edu}, {\tt\small neerajain@purdue.edu}.}%
\thanks{$^{2}$Jacob A. Siefert and Herschel C. Pangborn are with the Department of Mechanical Engineering, The Pennsylvania State University, University Park, PA 16802 USA, {\tt\small jas7031@psu.edu}, {\tt\small hcpangborn@psu.edu}.}%
}

\begin{document}

\maketitle
\thispagestyle{empty}
\pagestyle{empty}

\begin{abstract}

Control Co-Design (CCD) considers the coupled effects of both the plant and control parameters to optimize a system's closed-loop transient performance during the design stage. 
This paper presents a new method for CCD with guarantees on robustness to nondeterministic disturbances for all initial conditions within a specified region of operation. 
This is accomplished by calculating the reachable sets of a candidate closed-loop system directly within the optimization problem. 
Using this approach, the plant and control parameters are simultaneously chosen to shape these reachable sets to be robustly positive invariant and thus safe for all time. 
Compared to conventional approaches that perform the optimization for a single initial condition and an a priori chosen sequence of disturbances, the proposed set-based method avoids sensitivity to variations in the assumed design scenario. As a representative example, the proposed method is applied to an active suspension system. 

\end{abstract}



\section{Introduction}\label{sec-intro}

\emph{Motivation and Problem Definition:} Control Co-Design (CCD), also known as combined plant and control design, provides a promising tool for the design of closed-loop systems. By simultaneously choosing the coupled plant and control decision variables, the design space is broadened, resulting in systems with transient performance that may not be realized through a more conventional sequential design process \cite{garcia-sanz_control_2019}. The CCD problem is often formulated as the solution to an optimization program where the dynamics of the system are enforced through equality constraints and the objective function is evaluated based on the plant and control variables as well as the resulting trajectory of the system \cite{herber_nested_2018}. To formulate the optimization problem, the designer must supply an initial condition and sequence of exogenous signals a priori. This results in a system that is designed to perform as desired under the specific operating condition or scenario posed in the problem formulation. However, in practice there almost always exists uncertainty in the operating regime, exogenous signals, and modeled dynamics \cite{azad_control_2022}. In such cases, the CCD problem must be formulated to be robust to these uncertainties to ensure that the resulting closed-loop system will meet the performance objectives across the operating regime. 

\emph{Gaps in Literature:} While classical CCD approaches can be immediately adapted to evaluate multiple deterministic scenarios simultaneously \cite{allison_co-design_2014}, formulations that account for nondeterministic uncertainties result in a more numerically challenging optimization problem. 
Researchers have developed Robust Control Co-Design (RCCD) methods as a stochastic or min-max optimization problem to determine static plant variables with open-loop control signals \cite{cui_comparative_2019,azad_investigations_2023}. However, the resulting system is designed for an ideal control signal that may not be achieved online during operation. 
The reader is directed to \cite{azad_overview_2023} for a thorough discussion of the existing formulations for RCCD. 

For safety-critical systems, the RCCD problem must account for the worst-case scenario during the design stage. Existing approaches have solved the worst-case RCCD problem by employing well-studied robust control strategies within a nested optimization problem, including $\text{H}_\infty$ \cite{nash_combined_2020}, min-max MPC \cite{nash_robust_2021}, and tube-based MPC \cite{tsai_robust_2023}. While useful for guaranteeing the robustness of the RCCD system during online operation, these approaches suffer from the same drawbacks as the robust control strategies used in their design and operation, most notably 
conservatism in the system's safe operating region and large computational burden \cite{di_cairano_model_2012,yu_inherent_2014}. 

\emph{Contribution:} This paper presents a new method for RCCD that accounts for uncertainty using set-based reachability. Set-based methods for reachability analysis provide a powerful tool for determining the safety of systems with uncertainty in the dynamics, exogenous inputs, and initial conditions \cite{althoff_set_2021}.
By allowing the plant and control design parameters of the closed-loop system to change, and analyzing the reachable sets within an optimization problem, the proposed approach shapes the reachable sets containing all possible trajectories. We show how support functions can be used to translate complex reachable sets to scalar values for use in a cost function to favor candidate systems with all possible trajectories contracting quickly towards a reference value. Furthermore, we show how the same support functions can be used to build inequality constraints that ensure the system is robustly safe within a user-defined operating region, rather than the operating regime being evaluated after the design has been fixed. Using this robust set-based design approach, provably safe closed-loop systems can be designed using flexible control strategies that are not necessarily robust, and may therefore be more suitable for online operation. 
The only restrictions placed on the system and controller types that may be designed with the proposed method is that outer bounds on the closed-loop system's reachable set can be found. This includes systems with linear dynamics, nonlinear dynamics, hybrid dynamics, and optimal controllers \cite{althoff_set_2021,bird_set-based_2022}, 
thus providing a robust design tool suitable for many closed-loop systems. 

\emph{Outline:} In Sec. \ref{sec-prelims} we define notation and briefly discuss control co-design and set-based reachability analysis. In Sec. \ref{sec-RCCD} we present a new method for robust control co-design that optimizes the closed-loop system while considering all possible trajectories originating from an operating region. In Sec. \ref{sec-numEx} we provide an illustrative example followed by concluding remarks in Sec. \ref{sec-conclusions}. 
\section{Notation and Preliminaries}\label{sec-prelims}

The concatenation of two column vectors into a single column vector is denoted by $\left(x~y\right)=\left[x^T~y^T\right]^T$.
Sets are denoted by uppercase calligraphic letters, e.g., $\mathcal{Z}\subset{\rm I\!R}^{n}$. 
Given the sets $\mathcal{Z},\:\mathcal{W}\subset{\rm I\!R}^{n}$ and matrix $R\in\RN^{m\times n}$, the linear mapping of $\mathcal{Z}$ by $R$ is $R\mathcal{Z}=\{Rz~|~z\in\mathcal{Z}\}$ and
the Minkowski sum of $\mathcal{Z}$ and $\mathcal{W}$ is $\mathcal{Z}\oplus\mathcal{W}=\{z+w~|~z\in\mathcal{Z},\:w\in\mathcal{W}\}$. 
The support function of a set $\mathcal{Z}\subset{\rm I\!R}^{n}$ in a direction $l\in{\rm I\!R}^n$ is given by
\begin{equation}\label{eqn-supportFnZh}
  \rho_{\mathcal{Z}}(l)=\max\left\{l^Tz~\middle|~z\in\mathcal{Z}\right\}\:,
\end{equation}
and defines the supporting halfspace
\begin{equation}\label{eqn-supprtHalfZh}
    \mathcal{H}_{\mathcal{Z}}(l)=\left\{z\in{\rm I\!R}^{n}~\middle|~ l^Tz\leq\rho_{\mathcal{Z}}(l)\right\}\:,
\end{equation}
such that $\mathcal{Z}\subset\mathcal{H}_{\mathcal{Z}}(l)$ \cite{le_guernic_reachability_2010}. The template polyhedron of a set $\mathcal{Z}$ is defined as the intersection of halfspaces
\begin{equation}\label{eqn-templatePoly}
\lceil\mathcal{Z}\rceil_{\mathcal{L}}=\bigcap_{l\in\mathcal{L}}\mathcal{H}_{\mathcal{Z}}(l)\:,
\end{equation}
for a finite collection of directions $\mathcal{L}=\{l_1,\dots,l_k\}$. 
The template polyhedron of a set $\mathcal{Z}$ forms a convex over-approximation that is tight in the sense that the corresponding hyperplanes $\mathcal{H}_{\mathcal{Z}}(l)$ intersect the set $\mathcal{Z}$ for all $l\in\mathcal{L}$ and $\mathcal{Z}\subseteq\lceil\mathcal{Z}\rceil_{\mathcal{L}}$. Given an H-rep polytope 
\begin{equation}
    \mathcal{X}=\{x\in\RN^{n}~\vert~Hx\leq f\}\subset\RN^{n}\:,
\end{equation}
with $H\in\RN^{k \times n}$ and $f\in\RN^{k}$,
let $\mathcal{L}=\{h_1,\dots,h_k\}$ where $h_i\in\RN^{n}$ are the $k$ columns of $H^T$. Then it holds that $\lceil\mathcal{X}\rceil_{\mathcal{L}}=\mathcal{X}$ and $\mathcal{Z}\subseteq\mathcal{X}\iff\lceil\mathcal{Z}\rceil_{\mathcal{L}}\subseteq \mathcal{X}$. The containment of a set $\mathcal{Z}$ by the polyhedron $\mathcal{X}$ may be determined by evaluating the support functions of $\mathcal{Z}$ as
\begin{equation}
    \mathcal{Z}\subseteq\mathcal{X}\iff \rho_{\mathcal{Z}}(h_i)\leq f_i~\forall~h_i\in\mathcal{L}
\end{equation}
where $h_i$ and $f_i$ define the $k$ inequalities of $\mathcal{X}$ \cite{eaves_optimal_1982}. For ease of readability, given an H-rep polytope $\mathcal{X}$, the template polyhedron of a set $\mathcal{Z}$ in the normal directions of the faces of $\mathcal{X}$ is denoted by $\lceil\mathcal{Z}\rceil_{\mathcal{X}}$. 

\subsection{Plant and Control Co-Design}

Consider the closed-loop discrete-time dynamic system given by the difference equation
\begin{equation}\label{eqn-discreteDynamics}
    x(k+1)=f(x(k),\pi(x(k),p_c),v(k),p_p)\:,
\end{equation}
where $x(k)\in\RN^{n}$ is the system state and $v(k)\in\RN^{d}$ is a disturbance signal at time step $k\in\mathbb{Z}_{+}$. 
The plant design vector $p_p\in\RN^{p}$ captures parameters of the system dynamics \eqref{eqn-discreteDynamics} that may be decided by the designer.
The control signal $u(k)\in\RN^{m}$ at time step $k$ is given by the feedback control policy $u(k)=\pi(x(k),p_c):\RN^{n}\times\RN^{c}\mapsto\RN^{m}$ fully parameterized by the time-varying system state $x(k)$ and control design vector $p_c\in\RN^{c}$. For ease of readability, the difference equation \eqref{eqn-discreteDynamics} will be referenced as 
\begin{equation}\label{eqn-discreteDynamics-short}
    x(k+1)=f(x(k),v(k),p)\:,
\end{equation}
where $p=(p_p~p_c)\in\RN^{p+c}$ and the control policy is included within the dynamics $f(\cdot)$. 

Control co-design aims to simultaneously choose the plant design vector $p_p$ and control design vector $p_c$ to optimize the transient performance of the system \eqref{eqn-discreteDynamics}. This is conventionally accomplished by determining the co-designed system as the solution to the dynamic optimization problem
\begin{equation}\label{eqn-CCD-conventional}
\begin{split}
\min_{X,p}& \:\:m(N,p,x(N))+\sum_{k=0}^{N-1}l(k,p,x(k),v(k))\\
\text{s.t.}&\:\:x(k+1)=f(x(k),v(k),p)\:,\\
& \:\:g(k,x(k),v(k),p)\leq \mathbf{0}\:,\\
\text{for}&\:\:k=0,\dots,N\:,
\end{split}
\end{equation}
where $X=(x(0)~\cdots~x(N))$ is the trajectory of the system sampled over the finite-time horizon $k\in\{0,\dots,N\}$, $m(N,p,x(N))$ is the terminal cost and $l(k,p,x(k),v(k))$ is the running cost of the design. The inequality constraints $g(k,x(k),v(k),p)\leq \mathbf{0}$ are used to ensure that performance properties are met and that the design parameters $p$ are physically realizable. In this conventional formulation, the initial condition $x(0)$ and sequence of disturbances $V=(v(0)~\cdots~v(N-1))$ are constants chosen by the designer. 

The current state-of-the-art in the area of robust control co-design extends \eqref{eqn-CCD-conventional} to consider the worst-case realization of disturbances belonging to a bounded set $\mathcal{V}\subset\RN^{d}$ such that $v(k)\in\mathcal{V}~\forall~k\in\{0,\dots,N-1\}$. The robust formulation of the CCD problem is given by
\begin{equation}\label{eqn-CCD-conventional-robust}
\begin{split}
\min_{X,p,o}& \:\:o\\
\text{s.t.}&\:\:m(N,p,x(N))+\sum_{k=0}^{N-1}l(k,p,x(k),v(k))\leq o\:,\\
&\:\:x(k+1)=f(x(k),v(k),p)\:,\\
&\:\:g(k,x(k),v(k),p)\leq \mathbf{0}\:,\\
\text{for}&\:\:k=0,\dots,N,\:\:\text{and}\:\:\forall\:v(k)\in\mathcal{V}\:.
\end{split}
\end{equation}%
While the additional condition that the above must hold for all disturbances contained within the set $\mathcal{V}$ appears relatively benign, this makes the optimization program inherently a saddlepoint problem. 
The robust formulation of the CCD problem may be solved as a minimax, or bilevel program, where the lower level maximizes over the cost and constraint functions to find the worst-case possible sequence of disturbances to be considered in the design. 

\subsection{Reachable and Invariant Sets}

Reachability analysis consists of computing the set of states reachable by a dynamic system from a specified set for all admissible control inputs and disturbances \cite{althoff_set_2021}. Considering the dynamic system \eqref{eqn-discreteDynamics} with a set of states $\mathcal{R}(k)\subset\RN^{n}$ at time step $k$, the set of states reachable by the dynamic system \eqref{eqn-discreteDynamics} in one discrete time step for all possible disturbances within the set $\mathcal{V}\subset\RN^{d}$ is given by the forward reachable set as determined by the successor operator
\begin{equation}
    \texttt{Suc}(\mathcal{R}(k),\mathcal{V},p)=\left\{ f(x,v,p)~\middle\vert~x\in\mathcal{R}(k),\:v\in\mathcal{V}\right\}\:.
\end{equation}
Note that this expression differs from conventional definitions of the successor operator as it is now parametric with respect to the decision variables $p$.
The set of states reachable in time steps $k=0,\dots,N$ may be determined by successive application of the operator as $\mathcal{R}(k+1)=\texttt{Suc}\left(\mathcal{R}(k),\mathcal{V},p\right)$ from a specified initial set $\mathcal{R}(0)$. 

Given a safe subset of the state space $\overline{\mathcal{X}}\subset\RN^n$,
the set $\mathcal{O}\subseteq\overline{\mathcal{X}}$ is a robustly positive invariant set of the closed-loop system \eqref{eqn-discreteDynamics} for the disturbance set $\mathcal{V}$ if and only if $\texttt{\emph{Suc}}(\mathcal{O},\mathcal{V},p)\subseteq\mathcal{O}$, as all trajectories originating from the set $\mathcal{O}$ are fully contained within itself in a single time step \cite{blanchini_set_1999}. Thus $x(k)\in\mathcal{O}\implies x(k+i)\in\mathcal{O}\:\forall\:i\in\mathbb{Z}_{+}$ and therefore $x(k+i)\in\overline{\mathcal{X}}\:\forall\:i\in\mathbb{Z}_{+}$. Any closed-loop system satisfying this robustness condition is guaranteed to be safe for all time. In the following, we will address the problem of choosing $p$ such that this is guaranteed to be true. 

\section{Set-Based Robust Co-Design}\label{sec-RCCD}

This section  shows how reachability analysis may be leveraged to formulate the robust control co-design optimization problem to consider all possible responses of the system simultaneously. 
This is achieved by determining the reachable sets of a candidate system within the optimization problem. 
The design criteria of the proposed RCCD method is to simultaneously choose the plant parameters $p_p$ and control policy parameters $p_c$ such that for a given operating region $\mathcal{R}(0)$ and set of possible disturbances $\mathcal{V}$, the optimal closed-loop system 
\begin{enumerate}
    \item minimizes the size of a set of costs for the tube of all possible trajectories originating from the operating region $\mathcal{R}(0)$,
    \item minimizes the cost of the design parameters $p_p$ and $p_c$,
    \item robustly remains within a safe subset of the state space $\overline{\mathcal{X}}$ for all possible disturbances within the set $\mathcal{V}$ for all time,
    \item has a feedback control law that results in admissible actions belonging to the set $\overline{\mathcal{U}}$ for all possible states,
    \item the design parameters $p_p$ and $p_c$ are physically realizable.
\end{enumerate}
We first define the proposed RCCD optimization problem in terms of reachable sets, followed by a discussion of how each term satisfies the design criteria. We then propose a computationally tractable form of the problem and discuss its solution. 

\begin{prob}[Set-Based RCCD]
\label{prob-reachCCD-standard}
Given a safe subset of the state space $\overline{\mathcal{X}}\subset\RN^n$, operating region $\mathcal{R}(0)\subseteq\overline{\mathcal{X}}$, and set of admissible control actions $\overline{\mathcal{U}}\subseteq\RN^m$, the set-based RCCD of the closed-loop system \eqref{eqn-discreteDynamics} satisfying the design criteria is given by the solution to the optimization problem
\begin{subequations}\label{eqn-reachCCD-standard}
    \begin{align}
        \min_{p,\mathcal{R}}& \:\:m(N,p,\mathcal{R}(N))+\sum_{k=0}^{N-1}l(k,p,\mathcal{R}(k),\mathcal{V})\label{eqn-reachCCD-standard-cost}\\
        \text{s.t.}&\:\:\mathcal{R}(k+1)=\texttt{Suc}(\mathcal{R}(k),\mathcal{V},p)\:,\label{eqn-reachCCD-standard-reach}\\
        & \:\:\mathcal{R}(k+1)\subseteq\overline{\mathcal{X}}\:,\label{eqn-reachCCD-standard-safe}\\
        & \:\:\mathcal{R}(N)\subseteq\bigcup_{i=0}^{N-1}\mathcal{R}(i)\:,\label{eqn-reachCCD-standard-invariant}\\
        & \:\: \left\{  \pi(x,p_c) ~\middle\vert~ x\in\mathcal{R}(k) \right\}\subseteq \overline{\mathcal{U}}\:,\label{eqn-reachCCD-standard-control}\\
        & \:\:g(p)\leq \mathbf{0}\:,\label{eqn-reachCCD-standard-realizable}\\
        \text{for}&\:\:k=0,\dots,N\:.\label{eqn-reachCCD-standard-time}
    \end{align}
\end{subequations}
\end{prob}

The cost function \eqref{eqn-reachCCD-standard-cost} consists of a terminal cost $m(N,p,\mathcal{R}(N))$ penalizing the design variables $p$ and the final reachable set $\mathcal{R}(N)$, and a running cost $l(k,p,\mathcal{R}(k),\mathcal{V})$ as a function of the design variables, reachable sets, and disturbances at time steps $k=0,\dots,N-1$. The reachable sets are propagated forward in time through the constraint \eqref{eqn-reachCCD-standard-reach} originating from the operating region $\mathcal{R}(0)$ for all possible disturbances $\mathcal{V}$. The system is robustly safe when \eqref{eqn-reachCCD-standard-safe} is satisfied as all possible trajectories remain within the safe region $\overline{\mathcal{X}}$. Moreover, the solution is guaranteed to be safe for all time as the combination of constraints \eqref{eqn-reachCCD-standard-safe} and \eqref{eqn-reachCCD-standard-invariant} requires that the union of the reachable sets over the horizon $k=0,\dots,N-1$ is a robust positive invariant set. The set containment constraints \eqref{eqn-reachCCD-standard-control} ensure that the control policy results in an admissible input. 
The constraints \eqref{eqn-reachCCD-standard-realizable} can be used to ensure that the optimal design parameters $p$ are physically realizable. 

The optimization problem posed in \eqref{eqn-reachCCD-standard} is effective at addressing the shortcomings of the conventional CCD approach because it simultaneously considers all possible trajectories of the dynamic system originating from a user-defined operating region for all sequences of disturbance signals. Thus avoiding any sensitivity to the scenarios posed in \eqref{eqn-CCD-conventional} and \eqref{eqn-CCD-conventional-robust}. By leveraging notions of invariance properties of the reachable sets, the set-based optimization problem \eqref{eqn-reachCCD-standard} is able to provide a certificate of infinite-time safety while only analyzing over a finite-number of discrete time steps. However, solving such an optimization program with set-based objectives and constraints may, in general, be quite difficult. We will now present how \eqref{eqn-reachCCD-standard} can be posed in a way that is computationally tractable and can be computed using existing nonlinear optimization solvers. 

\subsection{Reachability Analysis}\label{sec-rCCD-reach}

This section shows how the reachable sets of a candidate system for design parameters $p$ may be characterized within the optimization problem by the set's support functions, a vector of scalars. But first, we state a few conditions on the sets used in Problem \ref{prob-reachCCD-standard} that will aid us in solving the optimization problem. 
\begin{ass}\label{ass-operatingIsHrep}
    The safe subset $\overline{\mathcal{X}}\subset\RN^n$, operating region $\mathcal{R}(0)\subseteq\overline{\mathcal{X}}$, and set of admissible control actions $\overline{\mathcal{U}}\subseteq\RN^m$ are convex sets given by the H-rep polytopes
    \begin{subequations}
    \begin{align}
        \overline{\mathcal{X}}=&\left\{x\in\RN^n~\middle\vert~H_x x\leq f_x\right\}\:,\\
        \mathcal{R}(0)=&\left\{x\in\RN^n~\middle\vert~H_r x\leq f_r\right\}\:,\\
        \overline{\mathcal{U}}=&\left\{u\in\RN^m~\middle\vert~H_u u\leq f_u\right\}\:,
    \end{align}
    \end{subequations}
    where $H_x\in\RN^{n_{hx}\times n}$, $f_x\in\RN^{n_{hx}}$, $H_r\in\RN^{n_{hr}\times n}$, $f_r\in\RN^{n_{hr}}$, $H_u\in\RN^{n_{hu}\times m}$, and $f_u\in\RN^{n_{hu}}$.
\end{ass}
\begin{ass}\label{ass-reachHasSupport}
    The reachable sets, $\mathcal{R}(k)$, produced by the successor operator in \eqref{eqn-reachCCD-standard-reach} and set of control actions
    \begin{equation}\label{eqn-setOfActions}
        \mathcal{U}(k)=\left\{  \pi(x,p_c) ~\middle\vert~ x\in\mathcal{R}(k) \right\}\:,
    \end{equation}
    have a defined support function. That is, for all $\mathcal{R}(k)\subset\RN^n$ and $\mathcal{U}(k)\subset\RN^m$, and directions $l_x\in\RN^n$ and $l_u\in\RN^m$, we can determine $\rho_{\mathcal{R}(k)}(l_x)$ and $\rho_{\mathcal{U}(k)}(l_u)$ given by \eqref{eqn-supportFnZh}. Furthermore, it is assumed that the closed-loop dynamics satisfy the common axioms required by the specific method used to propagate the successor operator forward in time, e.g., the boundedness and uniqueness of trajectories. 
\end{ass}

Given any candidate system with design parameters $p$, we find the reachable sets of the system originating from the operating region $\mathcal{R}(0)$ for the time horizon $k=0,\dots,N$ by iterating over the successor operator. 
Once the reachable sets $\mathcal{R}(k)\subset\RN^n$ and control actions $\mathcal{U}(k)$ for $k=0,\dots,N$ have been found, their template polyhedrons are computed for
\begin{subequations}\label{eqn-reachTemplates}
\begin{align}
    &\lceil\mathcal{R}(k)\rceil_{\mathcal{L}_h}\:\forall\:k\in\{1,\dots,N\}\:,\label{eqn-reachTemplates-L}\\
    &\lceil\mathcal{R}(k)\rceil_{\overline{\mathcal{X}}}\:\forall\:k\in\{1,\dots,N\}\:,\label{eqn-reachTemplates-X}\\
    &\lceil\mathcal{R}(N)\rceil_{\mathcal{R}(0)}\:,\label{eqn-reachTemplates-0}
\end{align}
\end{subequations}
and the template polyhedrons of the set of control actions as
\begin{subequations}\label{eqn-actionTemplates}
\begin{align}
    &\lceil\mathcal{U}(k)\rceil_{\mathcal{L}_h}\:\forall\:k\in\{1,\dots,N-1\}\:,\label{eqn-actionTemplates-box}\\
    &\lceil\mathcal{U}(k)\rceil_{\overline{\mathcal{U}}}\:\forall\:k\in\{1,\dots,N-1\}\label{eqn-actionTemplates-admis}\:,
\end{align}
\end{subequations}
where $\mathcal{L}_h=\{e_1,\dots,e_n,-e_1,\dots,-e_n\}$ and $e_i$ is the standard unit vector in each of the $i\in\{1,\dots,n\}$ cardinal directions. 
For ease of notation, for each time instance $k$ we define the vector of support functions \eqref{eqn-reachTemplates-L} as $\rho_{lx}^{k}$, \eqref{eqn-reachTemplates-X} as $\rho_{x}^{k}$, \eqref{eqn-reachTemplates-0} as $\rho_{r}$, \eqref{eqn-actionTemplates-box} as $\rho_{lu}^{k}$,
and those of \eqref{eqn-actionTemplates-admis} as $\rho_{u}^{k}$. The vectors of these values over the defined time horizon are simply denoted by dropping the superscript index, i.e., $\rho_{x}=\left(\rho_{x}^0~\cdots~\rho_{x}^N\right)$. 
This process is given by the function $[\tilde{\rho}_{lx},\tilde{\rho}_x,\tilde{\rho}_r,\tilde{\rho}_{lu},\tilde{\rho}_{u}]=\texttt{Reach}(\mathcal{R}(0),\mathcal{V},p,N,H_x,H_r,H_u)$ as described in Algorithm \ref{alg-reach}.

\begin{algorithm}[!ht]
\caption{Reachability analysis of candidate system.
}\label{alg-reach}
\algorithmicrequire{ Operating region $\mathcal{R}(0)$, disturbance set $\mathcal{V}$, design parameters $p$, time horizon $N$, normal directions of safe subset $H_x$, normal directions of operating region $H_r$, normal directions of admissible control inputs $H_u$}\\
\algorithmicensure{ Support functions $\tilde{\rho}_{lx}$, $\tilde{\rho}_x$, $\tilde{\rho}_r$, $\tilde{\rho}_{lu}$, $\tilde{\rho}_u$ defining upper-bounds of the template polyhedron \eqref{eqn-reachTemplates}-\eqref{eqn-actionTemplates}}
\begin{algorithmic}[1]
\State $\tilde{\mathcal{R}}(0)\gets\mathcal{R}(0)$, $\mathcal{L}_h\gets\{e_1,\dots,e_n,-e_1,\dots,-e_n\}$
\For {$k=1,\dots,N$}
\State Compute $\tilde{\mathcal{R}}(k)$ s.t. $\tilde{\mathcal{R}}(k)\supseteq\texttt{Suc}(\tilde{\mathcal{R}}(k-1),\mathcal{V},p)$
\State Compute $\tilde{\mathcal{U}}(k)$ s.t. $\tilde{\mathcal{U}}(k)\supseteq\left\{  \pi(x,p_c) ~\middle\vert~ x\in\tilde{\mathcal{R}}(k) \right\}$
\State Compute $\tilde{\rho}_{lx}^k$, $\tilde{\rho}_{x}^k$, $\tilde{\rho}_{lu}^k$, $\tilde{\rho}_{u}^k$ s.t. $\rho_{lx}^k\leq\tilde{\rho}_{l}^k$, $\rho_{x}^k\leq\tilde{\rho}_{x}^k$, $\rho_{lu}^k\leq\tilde{\rho}_{lu}^k$, $\rho_{u}^k\leq\tilde{\rho}_{u}^k$
\EndFor
\State Compute $\tilde{\rho}_{r}$ s.t. $\rho_{r}\leq\tilde{\rho}_{r}$
\end{algorithmic}
\end{algorithm}

\begin{rem}
    Note that no conditions have been placed on the system dynamics, disturbance set $\mathcal{V}$, or control policy $\pi(x,p_c)$. The only requirements are those stated in assumptions \ref{ass-operatingIsHrep} and \ref{ass-reachHasSupport}, namely that the operating region and constraint supersets are convex polyhedron and that we can find the support functions of the system's reachable sets. 
    However, the computational burden of iterating over the successor operator and sampling the support function is heavily dependent on these choices. E.g., for linear dynamics with reachable sets represented by symmetric convex polytopes, the analysis could be carried out using either H-rep polytopes or zonotopes. However, sampling the support function of an H-rep polytope requires solving a linear program while sampling the support function of a zonotope is done algebraically. 
\end{rem}

\begin{rem}
    Recall the definition of the template polyhedron \eqref{eqn-templatePoly} based on the value of the set's support function \eqref{eqn-supportFnZh} for a direction $l$. The computational burden of Algorithm \ref{alg-reach} may be reduced when the directions used in $\mathcal{L}_h$, $\overline{\mathcal{X}}$, and $\mathcal{R}(0)$ are shared---i.e., when $\overline{\mathcal{X}}$ and $\mathcal{R}(0)$ are hyperrectangles, \eqref{eqn-reachTemplates-L}-\eqref{eqn-reachTemplates-0} are given by the same values of the support function. 
\end{rem}

\begin{rem}
    Over-approximations in both the reachable sets $\mathcal{R}(k)$ and support functions \eqref{eqn-reachTemplates}-\eqref{eqn-actionTemplates}, denoted by tildes, that obey $\mathcal{R}(k)\subseteq\tilde{\mathcal{R}}(k)\subseteq\lceil\tilde{\mathcal{R}}(k)\rceil$ are allowed. Over-approximations may be used to reduce the computational burden of the analysis at the cost of increased conservatism in the resulting design. 
\end{rem}

\subsection{Cost Function}\label{sec-rCCD-costs}

This section presents a method for translating the sets of all trajectories of the candidate system to scalar values that may be used within the cost function. Rather than evaluating the transient performance of a candidate system based on a single trajectory, we penalize poor transient performance based on the size of the reachable set and the distance of its geometric center from a desired value.
While the most accurate way to evaluate the size of a set is by it's volume, determining a set's volume is computationally expensive, even for the most efficient set representations. Instead, we propose using the set's support functions to provide an over-approximation of the set's radius as follows. 

For the template polyhedron of the reachable set $\lceil\mathcal{R}(k)\rceil_{\mathcal{L}_h}$ and control actions $\lceil\mathcal{U}(k)\rceil_{\mathcal{L}_h}$ sampled in the cardinal directions $\mathcal{L}_h=\{e_1,\dots,e_n,-e_1,\dots,-e_n\}$ calculated in Alg. \ref{alg-reach}, the difference of the template polyhedra with the desired values $x_{ref}$ and $u_{ref}$ is given by
\begin{equation}\label{eqn-RandUError}
\begin{split}
    \mathcal{E}_{x}(k)&=\lceil\mathcal{R}(k)\rceil_{\mathcal{L}_h}\oplus\{-x_{ref}\}\:,\\
    \mathcal{E}_{u}(k)&=\lceil\mathcal{U}(k)\rceil_{\mathcal{L}_h}\oplus\{-u_{ref}\}\:.
\end{split}
\end{equation}
Given that the sets \eqref{eqn-RandUError} are defined by the Minkowski sum of a template polyhedron and a singleton, their support functions are given by 
\begin{equation}
\begin{split}
    \rho_{\mathcal{E}_x(k)}(l) &= \rho_{l}^k-l^{T}x_{ref}\:,\\
    \rho_{\mathcal{E}_u(k)}(l) &= \rho_{u}^k-l^{T}u_{ref}\:,
\end{split}
\end{equation}
for all $l\in\mathcal{L}_h$.
The optimal robust co-design satisfying the design criteria will have the sets \eqref{eqn-RandUError} as small as possible and close to the origin. 

Given the support functions $\rho_{\mathcal{E}(k)}(l)$, it is possible to calculate an estimate of the set's geometric center and radius as follows. 
The estimated geometric center of the set is given by
\begin{equation}\label{eqn-centerBound}
\begin{split}
    c(\mathcal{E}(k))=&\frac{\left(\rho_{\mathcal{E}(k)}(e_1)~\cdots~\rho_{\mathcal{E}(k)}(e_n)\right)}{2}\\
    &-\frac{\left(\rho_{\mathcal{E}(k)}(-e_1)~\cdots~\rho_{\mathcal{E}(k)}(-e_n)\right)}{2}\:.
\end{split}
\end{equation}
The distance that the set extends from its estimated geometric center in each direction is then given by
\begin{equation}\label{eqn-distancesForRad}
\begin{split}
    d(\mathcal{E}(k))=&\frac{\left(\rho_{\mathcal{E}(k)}(e_1)~\cdots~\rho_{\mathcal{E}(k)}(e_n)\right)}{2}\\
    &+\frac{\left(\rho_{\mathcal{E}(k)}(-e_1)~\cdots~\rho_{\mathcal{E}(k)}(-e_n)\right)}{2}\:.
\end{split}
\end{equation}
The $2-$norm radius of the set is then upper bounded by
\begin{equation}\label{eqn-radiusBound}
    R_2(\mathcal{E}(k))\leq\|d(\mathcal{E}(k))\|_2\:.
\end{equation}
This follows from the fact that $\|d(\mathcal{E}(k))\|_2$ is the radius of the set's smallest axis aligned hyperrectangle, and the ball defined by $\{x\in\RN^n~\vert~\|x-c(\mathcal{E}(k))\|_2\leq \|d(\mathcal{E}(k))\|_2\}$ is therefore guaranteed to contain the set $\mathcal{E}(k)$ as depicted in Figure \ref{fig-example_radiusAndCenter}.

Combining the upper bound of the set's radius and estimated geometric center, we now define a $2-$norm error metric as
\begin{equation}\label{eqn-2normMetric}
    \overline{R}_2(\mathcal{E}(k))=\|d(\mathcal{E}(k))\|_2+\|c(\mathcal{E}(k))\|_2\:,
\end{equation}
where the estimates of the set's geometric center, $c(\mathcal{E}(k))$, and distances from the center, $d(\mathcal{E}(k))$, are given by \eqref{eqn-centerBound} and \eqref{eqn-distancesForRad}, respectively. 
This error metric of an $n-$dimensional set requires sampling its support function $2n$ times. 
The scalar value given by \eqref{eqn-2normMetric} provides a measure of both how large the set is, as well as how far it lies from the origin as depicted in Figure \ref{fig-example_radiusAndCenter}.
We then formulate the running and terminal costs as
\begin{subequations}\label{eqn-costs-sets}
\begin{align}
    &m(N,p,\mathcal{R}(N)=m_p(p)+\gamma_1\overline{R}_2(\mathcal{E}_{x}(N))\:,\\
    &l(k,p,\mathcal{R}(k),\mathcal{V})=\gamma_2\overline{R}_2(\mathcal{E}_{x}(k))+\gamma_3\overline{R}_2(\mathcal{E}_{u}(k))\:,
\end{align}
\end{subequations}
where $m_p(p):\RN^{p+c}\mapsto\RN$ is any scalar-valued function penalizing the choice of the design parameters and can be chosen based on conventional metrics, e.g., cost or weight of components. The scalars $\gamma_1,\gamma_2,\gamma_3$ are used to  weigh the relative importance in reducing the cost of the sets of trajectories over the considered finite time horizon.
Using \eqref{eqn-costs-sets} in \eqref{eqn-reachCCD-standard} allows the designer to choose the relative importance between penalizing the size and location of the tube of all possible trajectories and making the control energy required to drive the system to these sets small. 

\begin{figure}[!htb]
     \centering
     \includegraphics[width=0.85\linewidth]{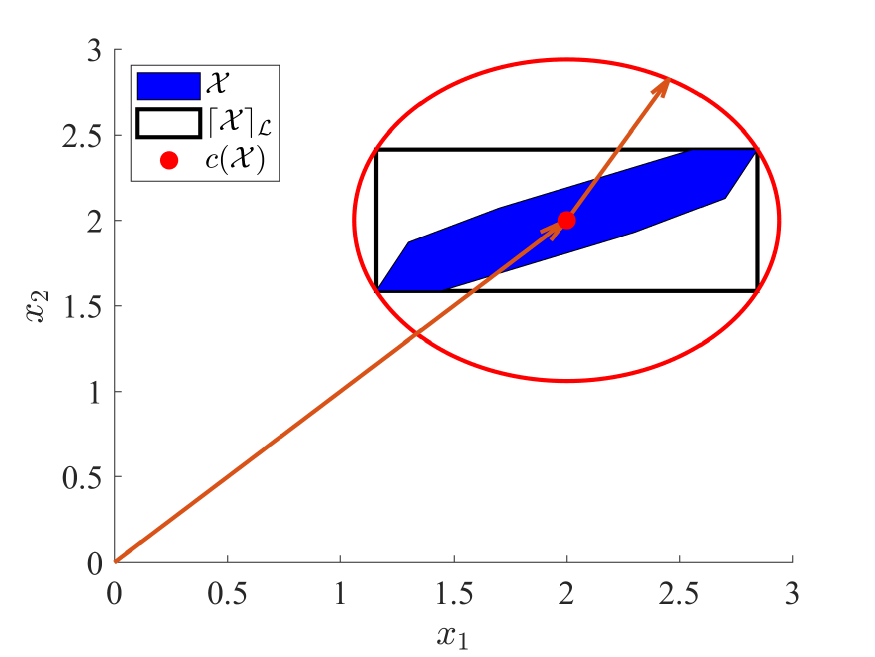}
     \caption{Example of approximating the $2-$norm radius and geometric center of a set using support functions. Vector from origin to geometric center and approximated radius depicted by arrows. 
        }
        \label{fig-example_radiusAndCenter}
\end{figure}

\begin{rem}
    While the desired states and control actions were defined as singletons, more complex costs may be used as long as they have a defined support function $\rho_{\mathcal{E}(k)}(l)$ over the sets of trajectories and control actions. 
\end{rem}

\begin{rem}
    When the dimensions of the reachable sets and set of control actions are weighed differently, these weights can be incorporated into the calculation of the proposed cost function \eqref{eqn-costs-sets} using linear transformations. Let $Q\in\RN^{q\times n}$ be a matrix of weights mapping the $n-$dimensional reachable set to the $q$ penalized values. The set \eqref{eqn-RandUError} for the cost of the states is then given by $\mathcal{E}_{x}(k)=\lceil Q\mathcal{R}(k)\rceil_{\mathcal{L}_h}\oplus\{-x_{nom}\}$.
\end{rem}

\subsection{Set-Containment Constraints}\label{sec-rCCD-setContain}

This section shows how sufficient conditions for the set-containment constraints given by \eqref{eqn-reachCCD-standard-safe} and \eqref{eqn-reachCCD-standard-invariant} may be imposed as inequality constraints using the reachable set's support functions. 
Methods for detecting the containment of two sets are computationally expensive and conditions for enforcing containment rarely exist \cite{sadraddini_linear_2019}. Once we consider the nonconvex union of multiple reachable sets as in \eqref{eqn-reachCCD-standard-invariant}, enforcing invariance becomes a difficult task. 
Instead, we introduce some conservatism by constraining the optimal response to have a tube of trajectories terminating within the operating region. Doing so guarantees  that the reachability analysis is complete and that the union of the reachable sets originating from the operating region is robustly positive invariant. 
\begin{prop}\label{prop-returnToR0}
    For any candidate system with reachable sets given by $\mathcal{R}(k+1)=\texttt{Suc}(\mathcal{R}(k),\mathcal{V},p)$ for $k=0,\dots,N-1$ and over-approximations $\mathcal{R}(k)\subseteq\tilde{\mathcal{R}}(k)$ with support functions $\rho^k_x\leq\tilde{\rho}^k_x$ defining the template polyhedron \eqref{eqn-reachTemplates}, the set containment constraints \eqref{eqn-reachCCD-standard-safe}-\eqref{eqn-reachCCD-standard-control} are satisfied if
    \begin{subequations}\label{eqn-prop-returnToR0}
        \begin{align}
            &\tilde{\rho}^k_x\leq f_x\:\forall\:k\in\{1,\dots,N\}\:,\label{eqn-prop-returnToR0-safe}\\
            &\tilde{\rho}^k_{u}\leq f_u\:\forall\:k\in\{0,\dots,N-1\}\:,\label{eqn-prop-returnToR0-control}\\
            &\tilde{\rho}^N_r\leq f_r\:.\label{eqn-prop-returnToR0-inv}
        \end{align}
    \end{subequations}
\end{prop}
\begin{proof}
    The satisfaction of the inequality constraints \eqref{eqn-prop-returnToR0} implying the satisfaction of \eqref{eqn-reachCCD-standard-safe}-\eqref{eqn-reachCCD-standard-control} follows from the requirement that the safe subset $\overline{\mathcal{X}}$, operating region $\mathcal{R}(0)$, and set of admissible control inputs $\overline{\mathcal{U}}$ are all H-rep polytopes and that the support functions are sampled in their normal directions. Thus it follows that $\rho^k_x\leq\tilde{\rho}^k_x\leq f_x$ implies $\mathcal{R}(k)\subseteq\tilde{\mathcal{R}}(k)\subseteq\overline{\mathcal{X}}$ and $\rho^k_u\leq\tilde{\rho}^k_u\leq f_u$ implies $\mathcal{U}(k)\subseteq\tilde{\mathcal{U}}(k)\subseteq\overline{\mathcal{U}}$. Similarly, $\rho^N_r\leq\tilde{\rho}^k_N\leq f_r$ implies $\mathcal{R}(N)\subseteq\tilde{\mathcal{R}}(N)\subseteq\mathcal{R}(0)$. Furthermore, $\mathcal{R}(0)\subseteq\bigcup_{i=0}^{N-1}\mathcal{R}(i)$ and therefore $\mathcal{R}(N)\subseteq\bigcup_{i=0}^{N-1}\mathcal{R}(i)$.
\end{proof}
\begin{rem}
    When the reachability analysis is exact, \eqref{eqn-prop-returnToR0-safe} is both necessary and sufficient for $\mathcal{R}(k)\subseteq\overline{\mathcal{X}}$. However, \eqref{eqn-prop-returnToR0-inv} is sufficient but not necessary to verify that the union of the reachable sets of the closed-loop system are positively invariant. This introduces some conservatism as $\mathcal{R}(0)$ and the finite time horizon $N$ are both constant parameters chosen by the designer. Alternatively, $\mathcal{R}(0)$ and $N$ may be considered as decision variables within the optimization problem; however, this further complicates the solution as it introduces a competing objective between maximizing the size of $\mathcal{R}(0)$ while minimizing the costs of the size of the reachable sets $\mathcal{R}(j)$ for $j>0$. An alternative approach is to find a convex robustly positive invariant set for each candidate system and use it in lieu of the operating region in \eqref{eqn-prop-returnToR0-inv}. However, finding these convex invariant sets may be computationally expensive \cite{blanchini_set_1999} and limits the types of systems that may be considered. 
\end{rem}

\subsection{Solving the Set-Based RCCD Problem}

We now present a form of the set-based RCCD problem that may be computed using existing nonlinear optimization solvers. This is given by
\begin{subequations}\label{eqn-reachCCD-solve}
    \begin{align}
    \min_{p,\tilde{\rho}}&=\begin{array}{l}
         m_p(p)+\gamma_1\overline{R}_2(\mathcal{E}_x(N))\\
        \:\:\:\:\:\:+\sum_{k=0}^{N-1}\gamma_2\overline{R}_{2}(\mathcal{E}_x(k))+\gamma_3\overline{R}_2(\mathcal{E}_u(k))
    \end{array}\label{eqn-reachCCD-solve-cost}\\
        \text{s.t.}&\:\:[\tilde{\rho}_l,\tilde{\rho}_x,\tilde{\rho}_r,\tilde{\rho}_u]=\texttt{Reach}(\mathcal{R}(0),\mathcal{V},p,N,H_x,H_r,H_u)\:,\label{eqn-reachCCD-solve-reach}\\
        & \:\:\tilde{\rho}_x\leq (f_x,\dots,f_x)\:,\label{eqn-reachCCD-solve-safe}\\
        & \:\:\tilde{\rho}_u\leq (f_u,\dots,f_u)\:,\label{eqn-reachCCD-solve-control}\\
        & \:\:\tilde{\rho}_r\leq f_r\:,\label{eqn-reachCCD-solve-invariant}\\
        & \:\:g(p)\leq \mathbf{0}\:.\label{eqn-reachCCD-solve-realizable}
    \end{align}
\end{subequations}%
The optimization problem \eqref{eqn-reachCCD-solve} consists of a nonlinear cost function depending on the choice $m_p(p)$ and 2-norm costs on the sets as defined in section \ref{sec-rCCD-costs}. 
The inequality constraints \eqref{eqn-reachCCD-solve-safe}-\eqref{eqn-reachCCD-solve-invariant} are linear and guarantee the safe operation of the system for all time as described in section \ref{sec-rCCD-setContain}. 
The reachability analysis \eqref{eqn-reachCCD-solve-reach} returns the support functions for the reachable set and set of control actions using Algorithm \ref{alg-reach} as discussed in section \ref{sec-rCCD-reach}. 
In this form, the feasible space of the design parameters is a subset of those given by \eqref{eqn-reachCCD-standard} due to the allowance of over-approximations in the reachability analysis and computation of template polyhedra, as well as the imposed sufficient condition for invariance as described in section \ref{sec-rCCD-setContain}. 
The additional computational burden required to solve this optimization problem compared to its nonrobust counterpart is dependent on the difficulty of the computation of the reachable sets, sampling of their support functions, and the inability to provide analytic gradients of the equality constraints \eqref{eqn-reachCCD-solve-reach}. 
\section{Numerical Example}\label{sec-numEx}

This section applies the proposed set-based RCCD method to the well-studied active suspension system \cite{allison_co-design_2014,sundarrajan_towards_2021}. The dynamic system consists of four states, $x=\begin{bmatrix}
    (z_{us}-z_0) &
    \dot{z}_{us} &
    (z_s-z_{us}) &
    \dot{z}_s
\end{bmatrix}^T$, where $z_{us}$ is the vertical position of the unsprung mass, $z_{s}$ is the vertical position of the sprung mass, and $z_0$ is the vertical position of the road, as depicted in Figure \ref{fig-activeSUSDiag}. 
The dynamics of the system are given by the linear state space model
\begin{subequations}\label{eqn-activeSUS-dynam}
\begin{align}
    \dot{x}=&A(p_p)x(t)+Bu(t)+E\dot{z}_0(t)\:,\\
    A(p_p)=&\begin{bmatrix}
        0 & 1 & 0 & 0 \\
        \frac{-k_t}{m_{us}} & \frac{-c_s}{m_{us}} & \frac{k_s}{m_{us}} & \frac{c_s}{m_{us}} \\
        0 & -1 & 0 & 1 \\
        0 & \frac{c_s}{m_{s}} & \frac{-k_s}{m_{s}} & \frac{-c_s}{m_{s}}
    \end{bmatrix}\:,\\
    B=&\begin{bmatrix}
        0 \\
        \frac{-1}{m_{us}} \\
        0 \\
        \frac{1}{m_{s}}
    \end{bmatrix}\:,\:E=\begin{bmatrix}
        -1 \\
        0 \\
        0 \\
        0 \\
    \end{bmatrix}\:,
\end{align}
\end{subequations}
where $k_t$ is the spring stiffness of the tire, $k_s$ is the spring stiffness of the suspension system, $c_s$ is the damping coefficient of the suspension system, $m_{us}$ is one fourth of the unsprung mass, and $m_s$ is one fourth of the sprung mass \cite{allison_co-design_2014,sundarrajan_towards_2021}. A zero-order hold is used to discretize the continuous dynamics for a sampling interval of $\Delta t = 0.01$ seconds. 

\begin{figure}[!htb]
     \centering
     \includegraphics[width=0.4\linewidth]{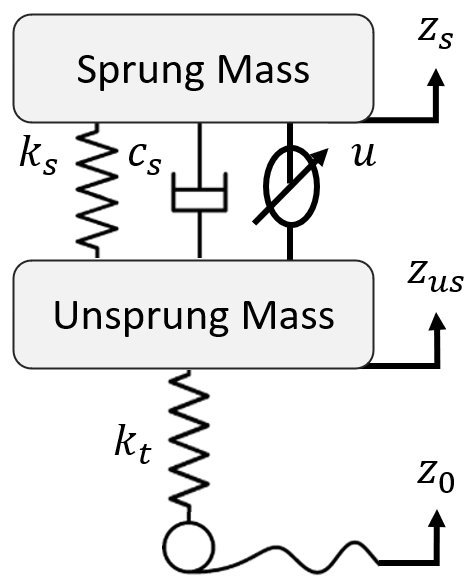}
     \caption{Diagram of the active suspension system \cite{allison_co-design_2014}. 
        }
        \label{fig-activeSUSDiag}
\end{figure}

The goal of the system design is to choose the stiffness and damping coefficients of the suspension system, $p_p=\begin{bmatrix}
    k_s & c_s
\end{bmatrix}^{T}\in\RN^2$, as well as a linear feedback gain defining the controller of the active system, $u(k)=\pi(x,p_c)=p_c^{T}x(k)$ for $p_c\in\RN^4$. The stiffness of the tire, sprung mass, and unsprung mass are all constants given by $k_t=232.5\times10^{3}$ N/m, $m_{s}=325$ kg, and $m_{us}=65$ kg \cite{allison_co-design_2014}. The decision variables $p_p$ and $p_c$ are constrained to belong to the intervals $10^4\leq k_s\leq 10^5$ N/m, $10^3\leq c_s \leq 10^4$ Ns/m, and $-10^6\leq p_{c,i}\leq 10^6\:\forall\:i\in\{1,\dots,4\}$. 

The system is designed to be robust to all possible sequences of changes in the road profile given by $\dot{z}_0\in\mathcal{V}=[-0.2,0.2]$ m/s. We design the system to be robust around the operating region defined by the intervals
\begin{equation}
    \mathcal{R}(0)=\begin{bmatrix}
        -0.25,0.25\\
        -0.75,0.75\\
        -0.25,0.25\\
        -0.75,0.75
    \end{bmatrix}\begin{matrix}
        \text{m} \\ \text{m/s} \\ \text{m} \\ \text{m/s}
    \end{matrix}\:,
\end{equation}
with safety constraints enforced on the distance between the sprung mass and unsprung mass, $x_3=z_s-z_{us}\in\overline{\mathcal{X}}$, given by
\begin{equation}
    -0.5\text{m}\leq z_s-z_{us}\leq 0.5 \text{m}\:,
\end{equation}
to avoid damaging the suspension system. The actuator force that can be applied by the control system is constrained to belong to the interval set $p_c^{T}x(k)\in\overline{\mathcal{U}}=[-4\times 10^3,4\times 10^3]$ N.

Following the costs chosen in \cite{allison_co-design_2014} to penalize poor transient performance, the sets \eqref{eqn-RandUError} used in the cost function are defined as
\begin{equation}
\begin{split}
    \mathcal{E}_{x}(k)&=\lceil Q\mathcal{R}(k)\rceil_{\mathcal{L}_h}\oplus\{0\}\:,\\
    \mathcal{E}_{u}(k)&=\lceil\mathcal{U}(k)\rceil_{\mathcal{L}_h}\oplus\{0\}\:,
\end{split}
\end{equation}
where 
\begin{equation}
    Q = \begin{bmatrix}
        10^5 & 0 & 0 & \\
        0 & 0.5\frac{c_s}{m_{s}} & 0.5\frac{-k_s}{m_{s}} & 0.5\frac{-c_s}{m_{s}}
    \end{bmatrix}\:,
\end{equation}
and the goal is to keep the control actions small and the trajectories of the system close to the origin. 
The linear transformation of the reachable sets $Q\mathcal{R}(k)$ in the calculation of the costs defines the different weights, $10^5$ and $0.5$, as well as a linear mapping from the states to the acceleration of the sprung mass $\ddot{z}_{s}$ as defined by \eqref{eqn-activeSUS-dynam} for the variables $p$. The cost on the control effort is introduced through the weighting in \eqref{eqn-reachCCD-solve-cost} given by $\gamma_1=\gamma_2=1$ and $\gamma_3=10^{-5}$. The design parameters are weighed as $m_p(p)=10^{-2}k_s+10^{-1}c_s$.

The reachability analysis within the optimization problem is evaluated over a time horizon of $t=[0,0.2]$ seconds requiring $N=20$ discrete time steps. For each candidate system with constant design parameters, the reachability analysis of the linear dynamics \eqref{eqn-activeSUS-dynam} is performed using the zonotope set representation with support functions evaluated using the methods described in \cite{girard_efficient_2008}. The optimization problem \eqref{eqn-reachCCD-solve} is solved using MATLAB's sequential quadratic program function $\texttt{fmincon}$. Using one core of a laptop with a 2.1 GHz Intel i7 processor and 16 GB of RAM, the optimization problem was solved in 4.93 seconds. The resulting reachable set of the RCCD system is shown in Figure \ref{fig-example_activeSUSreach-RCCD}. The reachable set of the robust design is compared to the plant parameters designed using a simultaneous CCD approach with an open loop control signal \cite{allison_co-design_2014}. To close the loop, an infinite time LQR feedback gain evaluated for the weights used in computing the open-loop control is applied. The reachable set of the simultaneous design with the LQR feedback policy is shown in Figure \ref{fig-example_activeSUSreach-LQR}. The performance of the two design strategies for the nominal initial condition and sequence of disturbances used in the simultaneous design are shown in Figure \ref{fig-example_activeSUSsingle}. 

\begin{figure}[!htb]
     \centering
     \begin{subfigure}[b]{0.4\textwidth}
         \centering
         \includegraphics[width=\textwidth]{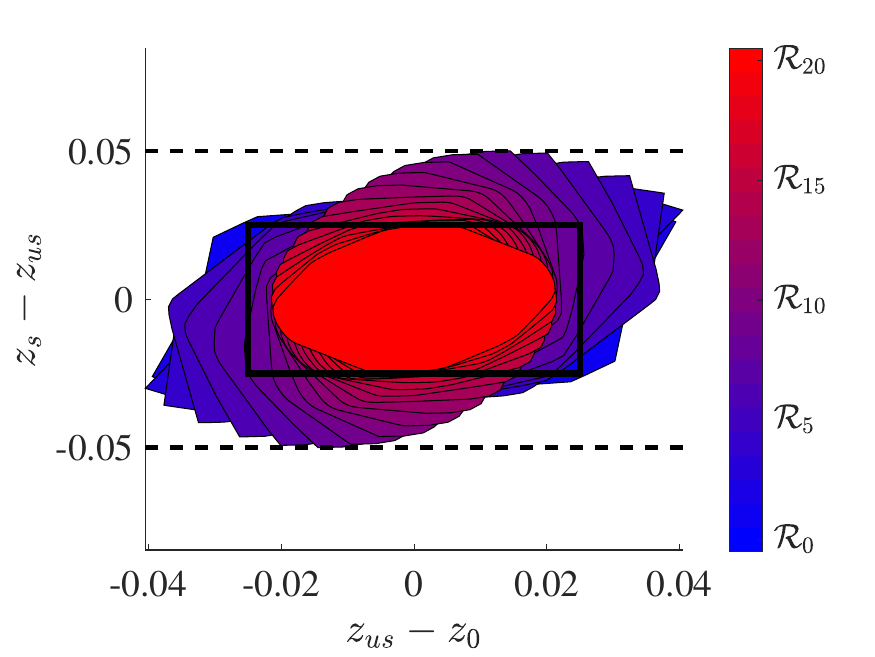}
         \caption{Set-based robust control co-design.}
         \label{fig-example_activeSUSreach-RCCD}
     \end{subfigure}
     \hfill
     \begin{subfigure}[b]{0.4\textwidth}
         \centering
         \includegraphics[width=\textwidth]{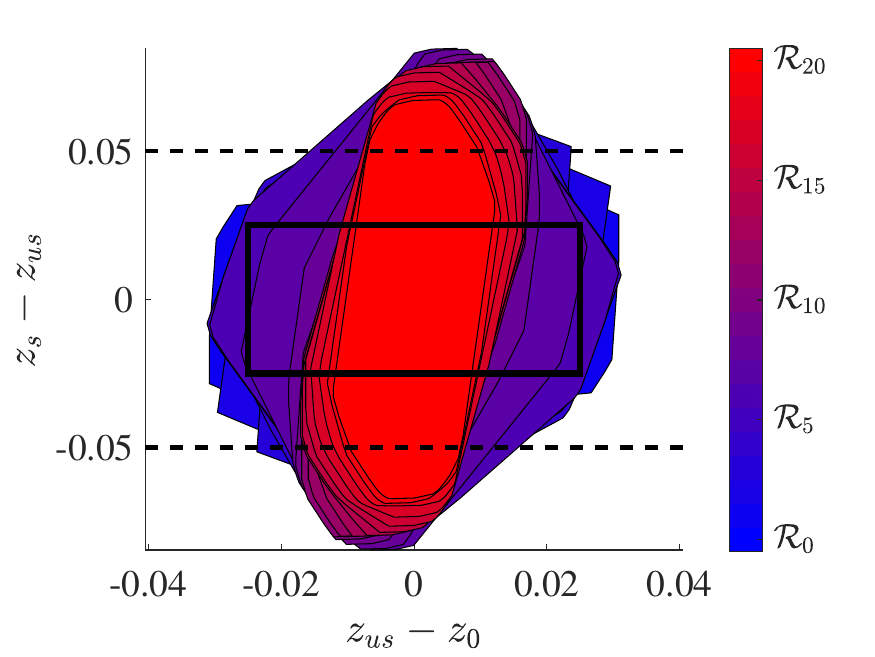}
         \caption{Simultaneous control co-design.}
         \label{fig-example_activeSUSreach-LQR}
     \end{subfigure}
     \hfill
        \caption{Projection of the two design strategies' reachable sets onto the first and third state dimensions. Safety constraints $\overline{\mathcal{X}}$ depicted by dashed lines and projection of the operating region depicted by bold black lines.}
        \label{fig-example_activeSUSreach}
\end{figure}



\begin{figure}[!htb]
     \centering
     \begin{subfigure}[b]{0.38\textwidth}
         \centering
         \includegraphics[width=\textwidth]{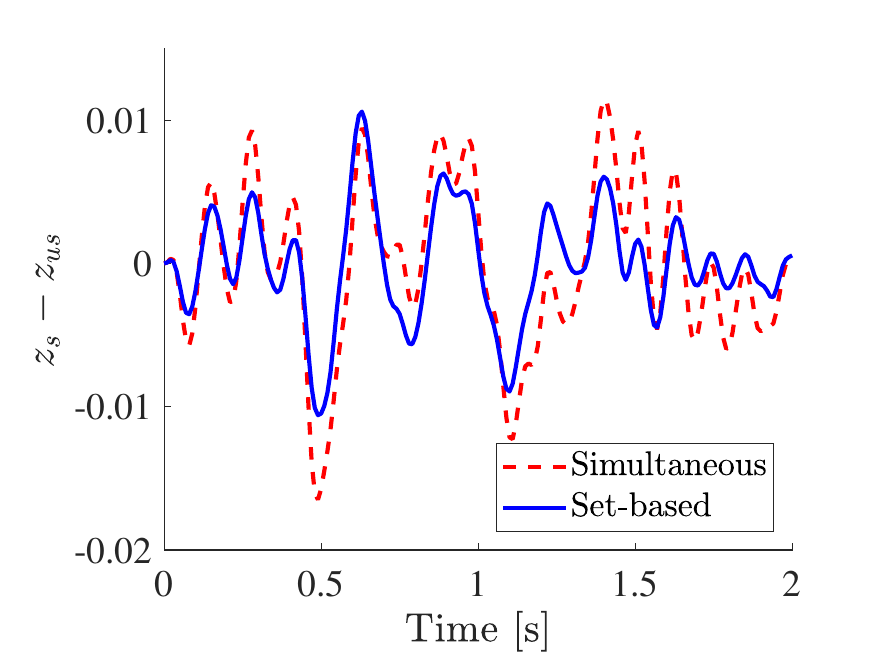}
         \caption{Trajectories of simultaneous and set-based designs.}
         \label{fig-example_activeSUSsingle-traj}
     \end{subfigure}
     \hfill
     \begin{subfigure}[b]{0.38\textwidth}
         \centering
         \includegraphics[width=\textwidth]{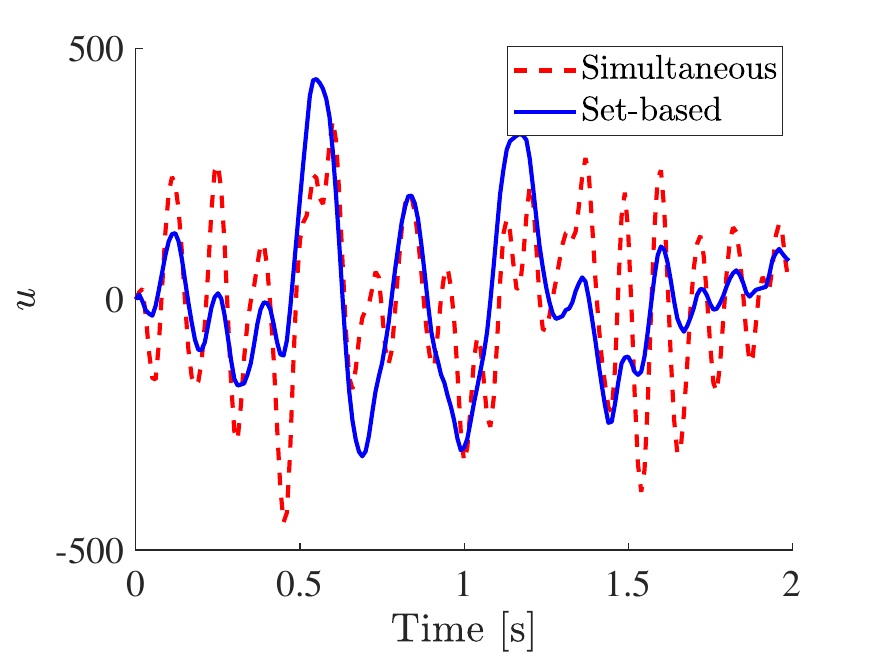}
         \caption{Control actions of simultaneous and set-based designs.}
         \label{fig-example_activeSUSsingle-U}
     \end{subfigure}
     \hfill
        \caption{Comparison of the transient performance of the two design strategies for the nominal scenario used in the simultaneous design.}
        \label{fig-example_activeSUSsingle}
\end{figure}

It can be seen from Figure \ref{fig-example_activeSUSreach-RCCD} that the reachable set of the system designed with the proposed RCCD method robustly satisfies the safety constraints. Furthermore, at the end of the considered time horizon the system has returned to the operating region, and therefore the reachability analysis is both complete, by design, and the closed-loop system is robustly positive invariant around the operating region. In the reachability analysis of the simultaneous design with an LQR controller, the system is designed for a single scenario and therefore not guaranteed to be safe, as can be seen in Figure \ref{fig-example_activeSUSreach-LQR} as the reachable set exceeds the safety constraint.
Comparing Figures \ref{fig-example_activeSUSreach-RCCD} and \ref{fig-example_activeSUSreach-LQR}, the set-based approach for RCCD shapes the reachable sets to extend further in the $z_{us}-z_0$ dimension where there are no constraints. This is achieved by increasing the stiffness and damping coefficients, as well as using a more aggressive feedback gain as shown in Table \ref{tab-designParams}. Examining the performance under the nominal scenario that the sequential design was optimized over, it can be seen in Figure \ref{fig-example_activeSUSsingle-traj} that the set-based approach provides a suitable design and control strategy while requiring a similar control effort as shown in Figure \ref{fig-example_activeSUSsingle-U}. However, once the system is subjected to a sequence of disturbances outside of the nominal scenario, only the set-based design is guaranteed to maintain safe operation. 

\begin{table}[!htb]
\caption{Plant and control parameters resulting from the two CCD methods. 
}
\begin{center}


  \begin{tabular}{c c c c c }
  \hline
  Method & $k_s$ & $c_s$ & $\max(|\rho_{lu}|)$& \\
  \hline
  

  Set-Based & $72064$ & $3888$ & $4000$ & \\
  Simultaneous & $23600$ & $1030$ & $2318$ & \\
  \hline
   & $p_{c,1}$& $p_{c,2}$& $p_{c,3}$& $p_{c,4}$\\
  \hline
  Set-Based & $-7922.6$ & $0$&$-50481$&$-3386.5$\\
  Simultaneous & $3121.2$ & $918.32$&$-5928.3$&$-1870.1$\\
  
  \hline
  \end{tabular}
\end{center}
\label{tab-designParams}
\end{table}
\section{Conclusions}\label{sec-conclusions}

This paper presents a new set-based method for robust control co-design. The proposed method leverages reachability analysis to optimize the closed-loop system over all possible trajectories originating from a user-defined operating region. 
It was shown how the support functions can be used to translate the sets of trajectories to scalar values in the cost function to design systems with tubes that contract quickly as well as enforce safety and invariant constraints. 
The proposed method was applied to the well-studied active suspension system and showed that the resulting design was safe for all time. 
Future work will focus on applying these methods to more complex systems and control laws.


\addtolength{\textheight}{-12cm}   









\bibliographystyle{IEEEtran}
\bibliography{bibTex_23.bib}

\end{document}